\documentclass[12pt,a4paper]{article}
\usepackage{tikz}
\usepackage{microtype}
\usepackage{geometry}                
\usepackage{amssymb,amsmath,amsthm}
\usepackage{color}
\usepackage{hyperref}
\usepackage{wrapfig}
\usepackage{url}
\usepackage{graphicx}
\usepackage{ifxetex}
\ifxetex
\usepackage{fontspec,xltxtra,xunicode}
\usepackage{unicode-math}
\defaultfontfeatures{Mapping=tex-text}
\newfontfamily\cyrillicfont{STIX Two Text}
\setromanfont[Mapping=tex-text]{STIX Two Text}
\usepackage{polyglossia}
\setdefaultlanguage{english}
\setotherlanguage{russian}
\else
\usepackage{mathptmx}
\fi

\newtheorem{theorem}{Theorem}
\newtheorem{proposition}[theorem]{Proposition}

\newtheorem{lemma}{Lemma}[theorem]
\theoremstyle{definition}
\newtheorem{definition}{Definition}
\theoremstyle{remark}
\newtheorem{remark}{Remark}
\newtheorem{question}{Question}

\newcommand{\cnd}{\mskip 2mu | \mskip 2mu}
\DeclareMathOperator{\KS}{\mathrm{C}\mskip 1mu}

\makeatletter
\newcommand*{\gl}{\nobreak\hskip\z@skip}
\DeclareRobustCommand*{\hyh}{\gl\hbox{-}\gl}

\title{Plain stopping time and\\ conditional complexities revisited}
\date{}
\author{Mikhail Andreev\thanks{IPONWEB, Berlin}, Gleb Posobin\thanks{National Research University Higher School of Economics, Moscow}, Alexander Shen\thanks{LIRMM CNRS / University of Montpellier, France. On leave from IITP RAS, Moscow. Supported by ANR-15-CE40-0016-01 RaCAF grant.}}
\begin{document}
\maketitle

\begin{abstract}
In this paper we analyze the notion of ``stopping time complexity'', the amount of information needed to specify when to stop while reading an infinite sequence. This notion was introduced by Vovk and Pavlovic~\cite{vovk-pavlovic}. It turns out that plain stopping time complexity of a binary string $x$ could be equivalently defined as (a)~the minimal plain complexity of a Turing machine that stops after reading $x$ on a one-directional input tape; (b)~the minimal plain complexity of an algorithm that enumerates a prefix-free set containing~$x$; (c)~the conditional complexity $\KS(x\cnd x*)$ where $x$ in the condition is understood as a prefix of an infinite binary sequence while the first $x$ is understood as a terminated binary string; (d) as a minimal upper semicomputable function $K$ such that each binary sequence has at most $2^n$ prefixes $z$ such that $K(z)<n$; (e) as $\max\KS^X(x)$ where $\KS^X(z)$ is plain Kolmogorov complexity of $z$ relative to oracle $X$ and the minimum is taken over all extensions $X$ of~$x$.

We also show that some of these equivalent definitions become non\hyh equivalent in the more general setting where the condition $y$ and the object $x$ may differ, and answer some open question from Chernov, Hutter and~Schmidhuber~\cite{chernov-et-al}.

\end{abstract}


\section{Introduction: stopping time complexity}\label{sec:intro}

Imagine that you explain to someone which exit on a long road she should take. You can just say ``$N$th exit''; for that you need $\log N$ bits. You may also say something like ``the first exit after the first bridge'', and this message has bounded length even if the bridge is very far away.\footnote{We do not allow, however, the description ``the last exit before the bridge'', since it uses information that is unavailable at the moment when we have to take the exit.}

More formally, consider a machine with one-directional read-only input tape that contains bits $x_0,x_1,\ldots x_n,\ldots$. We want to program the machine in such a way that it stops after reading bits $x_0,\ldots,x_{n-1}$ (and never sees $x_n$ and the subsequent bits).  Obviously, the complexity of this task does not depend on the values of $x_n,x_{n+1},\ldots$, because the machine never sees them, so this complexity should be a function of a bit string $x=x_0x_1\ldots x_{n-1}$. It can be called the ``stopping time complexity'' of $x$.

Such a notion was introduced recently by Vovk and Pavlovich~\cite{vovk-pavlovic}. In their paper an ``interactive'' version of stopping time complexity is considered where even ($x_{2n}$) and odd ($x_{2n+1}$) terms are considered differently, but this is just a special case, so we do not consider this setting. It turns out that the stopping time complexity is a special case of conditional Kolmogorov complexity with structured conditions. (In this paper we consider the plain version of stopping time complexity and postpone similar questions for prefix versions.)

The Kolmogorov complexity was introduced independently by Solomonoff, Kolmogorov, and Chaitin to measure the ``amount of information'' in a finite object (say, in a binary string).  One can also consider the \emph{conditional} version of complexity where some other object (a \emph{condition}) is given ``for free''. Later different versions of Kolmogorov complexity appeared (plain, prefix, a priori, monotone complexities). We assume that the reader is familiar with basic notions of algorithmic information theory, see, for example,~\cite{shen-survey} for a short introduction and~\cite{usv} for a detailed exposition.

For the the plain version of stopping time complexity we prove the equivalence between five different definitions (Section~\ref{sec:positive}). First, we show that it can be equivalently defined as (1)~the minimal plain complexity of a machine with one-way read-only input tape that stops after reading $x$, or (2)~the minimal enumeration complexity of a prefix-free set that contains $x$. Then we show how the stopping time complexity can be expressed in terms of plain conditional complexity that is monotone with respect to conditions. Namely, we prove that (3)~the stopping time complexity equals $\KS(x\cnd x*)$ where $x$ is used both as an object and a condition. Of course, according to standard definitions, the complexity $\KS(x\cnd x)$ is $O(1)$, but now we treat these two strings $x$ differently (and use a star in the notation to stress this). One may say that the topologies in the space of objects and the space of conditions are different. The first $x$ (object to be described) is considered as an isolated object (terminated string). The second $x$ (in the condition) is considered as a prefix of an infinite sequence. In~\cite[Section 6.3]{usv} this approach is described in general (see also~\cite{shenf0} for even more general setting); to make this paper self-contained, we give all necessary definitions for our special case. We call this version of complexity ``monotone-conditional complexity'' since this function is monotone with respect to the condition. Then we provide a characterization of stopping time complexity in 	quantitaive terms proving that (4)~stopping time complexity is the minimal upper semicomputable function satisfying some restrictions (no more than $2^n$ prefixes of any given sequence could have complexity at most $n$). Finally, we point out the connections with the relativized version $\KS^A(x)$ of plain complexity and prove that (5)~the stopping time complexity of a binary string $x$ is the maximal value of $\KS^A(x)$ for all oracles (infinite bit sequences) $A$ that have prefix $x$.

Having such a robust definition for plain stopping time complexity, one may ask whether similar characterizations can be obtained for a more general notion of $\KS(y\cnd x*)$ where $x$ and $y$ are arbitrary strings. Unfortunately, here the situation is much worse, and we prove mostly negative results (Section~\ref{sec:neg}). We show that while $\KS(y\cnd x*)$ can be defined as a minimal plain complexity of a prefix-stable program that maps $x$ to $y$ (Theorem~\ref{thm:monot-plain-char}), it cannot be defined as a minimal plain complexity of a prefix-free program that maps some prefix of $x$ to $y$ (Theorem~\ref{thm:stable-not-free}; this result answers a question posed in~\cite{chernov-et-al}). Then we show that the attempt to define $\KS(y\cnd x*)$ by quantitative restrictions also fails: we get another function that may be up to two times less (Theorem~\ref{thm:no-quant-char}).

\section{Equivalent definitions}\label{sec:positive} 

\subsection{Machines and prefix-free sets}\label{subsec:mach-prefix}

Consider a Turing machine $M$ that has one-directional read-only input tape with binary alphabet, and a work tape with arbitrary alphabet (or many work tapes). Let $x$ be a binary string. We say that $M$ \emph{stops at $x$} if $M$, being started with the input tape $x$ (and an empty work tape, as usual), reads all the bits of $x$ and stops without trying to read more bits. (We assume that initially the input head is on the left of $x$, so it needs to move right before seeing the first bit of $x$.) For a given $x$, we may consider the \emph{minimal plain Kolmogorov complexity of a machine $M$ that stops at $x$}. This quantity is independent (up to $O(1)$-additive term) of the details of the definition (work tape alphabet, number of work tapes, etc.) since computable conversion algorithms exist, and a computable transformation may increase complexity only by $O(1)$. 

\begin{definition}\label{def:plain-stopping}
We call this quantity the \emph{plain stopping time complexity} of $x$.
\end{definition}

Here is a machine-independent equivalent characterization of the plain stopping time complexity.

\begin{theorem}\label{thm:mach-prefix}
Plain stopping time complexity of $x$ equals \textup(up to an $O(1)$ additive term\textup) the minimal complexity of a program that enumerates some prefix-free set containing~$x$.
\end{theorem}

(A set of strings is called \emph{prefix-free} if it does not contain a string and its proper prefix at the same time.)

\begin{proof}
One direction is simple: For a Turing machine $M$ of the type described the set $\{x\colon \text{$M$ stops at $x$}\}$ is enumerable (we may simulate all runs) and prefix-free (if $M$ stops at some $x$, then for every extension $y$ of $x$ the machine $M$ will behave in the same way on tapes $x$ and $y$, so $M$ with input $y$ stops after reading $x$ and never reads the rest of $y$). This computable conversion (of a machine into an enumeration program) increases complexity at most by $O(1)$.

The other direction is a bit more complicated. Imagine that we have a program that enumerates some prefix-free set $U$ of strings. How can we construct a machine that stops exactly at the strings in $U$? Initially no bits of $x$ are read. Enumerating $U$, we wait until some element $u$ of $U$ appears. (If this never happens, the machine never stops, and this is OK.) If $u$ is empty, machine stops. In this case $U$ cannot contain non-empty strings (being prefix-free), so the machine's behavior is correct. If $u$ is not empty, we know that empty string is not in $U$ (since $U$ is prefix-free), so we may read the first bit of $x$ without any risk, and get some one-bit string $v$. Then we wait until $v$ or some extension of $v$ appears in $U$ (it may have already happened if $u$ is an extension of $v$). If $v$ itself appears, the computation stops; if a proper extension of $v$ appears, then $v$ is not in $U$ and we can safely read the next bit, etc. It is easy to check that indeed this machine stops at some $x$ if and only if $x$ belongs to $U$.
\end{proof}

\subsection{Monotone-conditional complexity}\label{subsec:monot-cond}

In this section we show how the stopping time complexity can be obtained as a special case of some general scheme~\cite{shenf0,us,usv}. This scheme can be used to define different versions of Kolmogorov complexity. We consider \emph{decompressors}, called also \emph{description modes}. In our case decompressor is a subset $D$ of the set
$$
(\text{descriptions})\times(\text{conditions})\times(\text{objects}).
$$
Here descriptions, conditions, and objects are binary strings. If $(p,x,y)\in D$, we say that \emph{$p$ is a description of $y$ given $x$ as condition},. We define the \emph{conditional complexity of $y$ given $x$} (with respect to the description mode $D$) as the length of the shortest description. The different versions of complexity correspond to different topologies on the spaces involved, and imply different restrictions on description modes. This is explained in~\cite{us} or \cite[Chapter 6]{usv}, and we do not go into technical details here. Let us mention only that descriptions and objects can be considered as isolated entities (terminated strings, natural numbers) or prefixes of an infinite sequence (extension of a string provides more information than the string itself). In this way we get four classical versions of complexity:
\begin{center}
\begin{tabular}{|c||c|c|}
\hline
     & isolated descriptions & descriptions as prefixes\\
\hline\hline
  isolated objects & plain complexity & prefix complexity\\
\hline
  objects as prefixes & decision complexity & monotone complexity\\
\hline
\end{tabular}
\end{center}
As noted in~\cite{usv}, one can also consider different structures on the condition space, thus getting eight versions of complexity instead of four in the table. In this paper we use only one of them: objects and descriptions are isolated objects, and conditions are considered as prefixes. (Vovk and Pavlovic~\cite{vovk-pavlovic} consider also another version of stopping time complexity that corresponds to the other topology on the description space, but we do not consider these versions now.)

To make this paper self-contained, let us give the definitions tailored to the special case we consider (the plain version of monotone-conditional complexity). The general scheme of a complexity definition reduces in this case to the following definition. Consider a set $D$ of triples $(p,x,y)$ where $p,x,y$ are binary strings. If $(p,x,y)$ is in $D$, we say \emph{that $p$ is a description of~$y$ with condition~$x$}. The set $D$ should satisfy the following requirements:
\begin{itemize}
\item $D$ is (computably) enumerable;
\item for every $p$ and $x$ there exists at most one $y$ such that $(p,x,y)\in D$;
\item if $(p,x,y)\in D$ and $x$ is a prefix of some $x'$, then $(p,x',y)\in D$.
\end{itemize}
Sets that satisfy these requirements are called \emph{description modes}. The last requirement reflects the idea that $x$ is considered as a known prefix of a yet unknown infinite sequence; if $x'$ extends $x$, then $x'$ contains more information than~$x$ and can be used instead of~$x$. To stress this kind of monotonicity, we use $*$ in the notation suggested by the following definition.

\begin{definition}\label{def:plain-monot-cond}
For a given description mode $D$, we define the function
$$
\KS_D(y\cnd x*)=\min\{ |p|\colon (p,x,y)\in D\} 
$$
and call it \emph{monotone-conditional complexity of $y$ with condition $x$ with respect to description mode $D$}.
\end{definition}
By definition, if $x$ is a prefix of some $x'$, the same description can be used, so $\KS_D(y\cnd x')\le\KS_D(y\cnd x)$. Therefore, this function is indeed monotone with respect to the condition in a natural sense.

One could also use a name \emph{plain monotone-conditional complexity} to distinguish this notion from prefix monotone-conditional complexity that can be defined in a similar way by adding the monotonicity restriction along the $p$-coordinate.

\begin{proposition}[Solomonoff--Kolmogorov's optimality theorem]\label{prop:plain-version}
There exists a description mode $D$ that makes $\KS_D$ minimal up to $O(1)$ additive term in the class of all functions $\KS_{D'}$ for all description modes $D'$.
\end{proposition}

\begin{proof}
As usual, we first note that description modes can be effectively enumerated. This enumeration is obtained as follows. We generate all enumerable sets of triples and then modify them in such a way that the modified set becomes a description mode and is left unchanged if it already were a description mode. Namely, when a triple $(p,x,y)$ appears in the enumeration, we add this triple and all triples $(p,x',y)$ for all extensions $x'$ of $x$, unless the second condition is violated after that; in the latter case we ignore $(p,x,y)$.  

Let $U_n$ be the $n$th set in this enumeration. The optimal set $U$ can be constructed~as
$$
 U=\{ (0^n1p,x,y) \colon (p,x,y)\in U_n\};
$$
the standard argument shows that $\KS_U \le \KS_{U_n}\!+\,n+1$ as required.
\end{proof}

\begin{definition}\label{def:plain-version}
Fix some optimal description mode $D$ provided by Proposition~\ref{prop:plain-version}. The function $\KS_D(y\cnd x*)$ is denoted by $\KS(y\cnd x*)$  and called the (plain) \emph{monotone-conditional complexity of $y$ given $x$}, or the (plain) \emph{conditional complexity of $y$ given $x$ as a prefix}.
\end{definition}

If we omit the third requirement for description modes, we get the standard conditional complexity $\KS(y\cnd x)$ in the same way. The notation we use (placing $*$ after the condition) follows~\cite{chernov-et-al} though a different version of monotone-conditional complexity is considered there. In general, $\KS(y\cnd x*)$ is greater than the standard conditional complexity $\KS(y\cnd x)$ since we have more requirements for the description modes. One may say also that the condition now is weaker than in $\KS(y\cnd x)$ since we do not know where $x$ terminates. It is easy to show that the difference is bounded by $O(\log |x|)$, since we need at most $O(\log |x|)$ bits to specify how many bits should be read in the condition~$x$. Difference of this order is possible: for example, $\KS(n\cnd 0^n)=O(1)$, but $\KS(n\cnd 0^n*)=\KS(n)+O(1)$ (the condition $0^n$ is a prefix of a computable sequence $000\!\ldots$, so it does not help).

The following simple result shows that the plain stopping time complexity (Definition~\ref{def:plain-stopping}) is a special case of this definition when $x=y$ (so we do not need a separate notation for the stopping time complexity).

\begin{theorem}\label{thm:plain-stopping}
The complexity $\KS(x\cnd x*)$ is equal \textup(up to $O(1)$ additive term\textup) to the plain stopping time complexity of~ $x$.
\end{theorem}

\begin{proof}
Let $D$ be a description mode. Then for every $p$ we may consider the set $S_p$ of $x$ such that $(p,x,x)\in D$. This set is prefix-free: if $(p,x,x)$ and $(p,x',x')$ belong to $D$ and $x$ is a prefix of $x'$, then $(p,x',x)\in D$ according to the third condition, and then $x=x'$ according to the second condition. The algorithm enumerating $S_p$ can be constructed effectively if $p$ is known, so its complexity is bounded by the length of $p$ (plus $O(1)$, as usual). Choosing the shortest $p$ such that $(p,x,x)\in D$, we conclude that the minimal complexity of an algorithm enumerating a prefix-free set containing $x$ does not exceed $\KS(x\cnd x*)+O(1)$ 

Going in the other direction, consider an optimal decompressor $U(\cdot)$ that defines the (plain Kolmogorov) complexity of programs enumerating sets of strings. A standard trimming argument shows that we may modify $U$ in such a way that all algorithms $U(p)$ enumerate only prefix-free sets of strings (not changing the sets there were already prefix-free). Then consider a set $D$ of triples
$$
 (p,x,y)\in D \Leftrightarrow \text{$y$ is a prefix of $x$ and $y$ is enumerated by $U(p)$.}
$$
This set is obviously enumerable; the second requirement is satisfied since $D(p)$ enumerates a prefix-free set; the third requirement is true by construction, so $D$ is a description mode. If $p$ is the shortest description of a program that enumerates a set containing $x$, then $(p,x,x)\in D$, so $\KS_D(x\cnd x*)\le |p|$. Switching to the optimal desciption mode, we get a similar inequality with $O(1)$ additive term, as required.
\end{proof}

Another simple observation shows that indeed this complexity is the \emph{stopping time} complexity.

\begin{proposition}\label{prop:length}
 If $x$ has length $n$, then
$\KS(x\cnd x*)=\KS(n\cnd x*)+O(1).$
\end{proposition}
\begin{proof}
If $D$ is the optimal description mode used to define $\KS(y\cnd x*)$, we may consider a new set 
  $D'=\{ (p,u,|x|) \colon (p,u,x)\in D\}$
that also is a description mode, and then note that $\KS_{D'}(|x|\cnd x*)\le \KS_D(x\cnd x*)$. For the other direction, we consider 
  $D'=\{(p,u,z)\colon \exists n\, [(p,u,n)\in D, |u|\ge n, \text{ and $z$ = ($n$-bit prefix of $u$)}] \}$.
\end{proof}

\begin{remark}
If $a_0a_1a_2\ldots$ is a computable sequence, then $$\KS(a_0\ldots a_{n-1}\cnd a_0\ldots a_{n-1}*)=\KS(n\cnd a_0\ldots a_{n-1}*)=\KS(n)$$
with $O(1)$-precision (the constant depends on the computable sequence, but not on $n$), so the stopping time complexity can be considered as a generalization of the plain complexity (of a natural number $n$).
\end{remark}

\subsection{Quantitative characterization}\label{subsec:quant}

There is a well known characterization (see, e.g., \cite[Section 1.1, Theorem 8]{us}) for plain complexity in terms of upper semicomputable functions that satisfy some properties. Recall that a function is called \emph{upper semicomputable} if it is a pointwise limit of a decreasing sequence of uniformly computable total functions. (Now we need this notion for integer-valued functions; in this case we may assume without loss of generality that these computable functions are also integer-valued; in general case one needs to consider rational-valued functions.) An equivalent definition of a semicomputable natural-valued function $S(x)$ requires the set $\{(n,x)\colon S(x)<n\}$ to be enumerable.

  Plain complexity function $\KS(x)$ is upper semicomputable; we know also that 
$$
\#\{ x\colon \KS(x)<n\} < 2^n \eqno(*)
$$
since there are less than $2^n$ programs of length less than $n$. The characterization that we mentioned says that there exist a minimal (up to $O(1)$ additive term) upper semicomputable function that satisfies~$(*)$, and it coincides with plain complexity function with $O(1)$-precision.

It turns out that this characterization can be generalized to plain stopping time complexity (though the proof becomes more involved). Consider upper semicomputable functions $S(x)$ on strings that have the following property: \emph{for each infinite binary sequence $\alpha$ and for each $n$ there exists less than $2^n$ prefixes $x$ of $\alpha$ such that $S(x)<n$}. The following statement is true (it appeared as Theorem 18 in the extended version of Vovk--Pavlovic's paper~\cite{vovk-pavlovic}).

\begin{theorem}\label{thm:cardinality}
There exist a minimal \textup(up to $O(1)$ additive term\textup) function in this class; it coincides with the plain stopping time complexity $\KS(x\cnd x*)$ with $O(1)$-precision.
\end{theorem}

\begin{proof}
The easy part is to show that $\KS(x\cnd x*)$ belongs to the class. It is upper semicomputable, since in general the function $\KS(x\cnd y*)$ is upper semicomputable (enumerating the set $D$ of triples, we get better and better upper bounds, finally reaching the limit value).

Let $\alpha$ be some infinite sequence. There are less than $2^n$ algorithms of complexity less than $n$ enumerating prefix-free sets, and each of this prefix-free sets may contain at most one prefix of $\alpha$. So the second condition is also true.

In the other direction we use some online (interactive) version of Dilworth theorem (saying that a partially ordered finite set where maximal chain is of length at most $k$ can be partitioned into $k$ antichains) where the set is a growing subset of the full binary tree and splitting into antichains should be performed at each stage (and cannot be changed later). The exact statement is as follows.

Consider a game with two players. Alice and Bob alternate. Alice may at each move (irreversibly) mark a vertex of a full binary tree. The restriction is that each infinite branch should contain at most $k$ marked vertices. Bob replies by assigning a color from $1,\ldots,k$ to the newly marked vertex. No vertices of the same color should be comparable (be on the same branch). The colors cannot be changed after they are assigned. Bob loses if he is unable to assign color at some stage (not violating the rules).
\begin{lemma}\label{lem:gleb}
Bob has a computable strategy that prevents him from losing.
\end{lemma}

\begin{proof}[Proof of Lemma~\ref{lem:gleb}]

This lemma can be proven in different ways. In the extended version of Vovk--Pavlovic's paper~\cite{vovk-pavlovic} the following simple strategy is suggested: Bob assigns the first available color. In other terms, for a new vertex $x$ Bob chooses the first color that is not used for any vertex comparable with $x$. One needs to check that $k$ colors are always enough. It is not immediately obvious, since more than $k$ vertices could be comparable with $x$ (being its descendants, for example). However, we may note that during the process:
\begin{itemize}
\item Colors of comparable vertices are different. (By construction.)
\item If a vertex $x$ gets color $i$, then each smaller color is used either for a predecessor of $x$ or for a descendant of $x$. (By construction.)
\item If $x$ is a vertex (colored or not),  $T_x$ is the set of colors used in the subtree rooted at $x$, and $P_x$ is the set of colors used on the path to $x$ (not including $x$), then $T_x$ and $P_x$ are disjoint and $T_x$ is the initial segment in the complement to $P_x$. (Indeed, the disjointness is mentioned above. If $y$ appears in $T_x$, then all smaller colors appear either below $y$ (therefore in $P_x$ or in $T_x$), or above $y$ (therefore in $T_x$).
\item The sets $T_{x0}$ and $T_{x1}$ for two brother vertices $x0$ and $x1$ are comparable with respect to inclusion. (Indeed, they are two initial segments of the same ordered set, the complement to $P_{x0}$ or $P_{x1}$; note that $P_{x0}=P_{x1}$.)
\item For each $x$ the total number of colors used in $T_x$ is the minimal possible, i.e., this number equals the maximal number of marked vertices on some path in $T_x$. (Induction using the previous property.)
\end{itemize}

The last property implies that Bob never uses more than $k$ colors, since by assumption the total number of marked vertices on one path is at most $k$.

There is a different description of the winning strategy for Bob (we provide it since it somehow explains why the previous argument works). At every stage, for each vertex $x$ we consider the \emph{marked rank} of $x$, the maximal number of marked vertices on some path (in $x$-subtree) that starts at $x$. By assumption we know that the marked rank of the root never exceeds $k$. Denoting the marked rank of $x$ by $r(x)$, we may write the recursive definition:
$$
r(x)=
\begin{cases}
\max(r(x0),r(x1)), \text{if $x$ is not marked};\\
\max(r(x0),r(x1))+1, \text{if $x$ is marked}.
\end{cases}
$$
(To complete this definition, we should add that $r(x)=0$ if the subtree rooted at $x$ has no marked vertices.)

On the other hand, for each vertex $x$ we consider the number of different colors used in the subtree rooted at $x$, and denote it by $c(x)$.  Let us denote by $C(x)$ the set of these colors, so $c(x)=\#C(x)$. We can write a similar recursive definition for $C(x)$:
$$
C(x)=
\begin{cases}
C(x0)\cup C(x1), \text{if $x$ is not marked};\\
C(x0)\cup C(x1)+(\text{the color of $x$)}, \text{if $x$ is marked}.
\end{cases}
$$
We use the sign $+$ in the last line, because the color of $x$ cannot be in $C(x0)$ or $C(x1)$ due to our requirements.

The game rules imply that $c(x)\ge r(x)$, since for every branch all the marked vertices on this branch should have different colors. Bob strategy is to maintain the invariant relation $c(x)=r(x)$, i.e., Bob uses the minimal possible number of colors for every subtree. We denote this invariant relation by~(I). If he manages to maintain it, he does not need more than $k$ colors, since by assumption $r(\Lambda)$ never exceeds $k$, where $\Lambda$ is the root (the empty string). But how can Bob maintain (I)?

Let us start with the following remark. Assume that (I)~holds. Then for every vertex $x$ the sets $C(x0)$ and $C(x1)$ are comparable (one of them is a subset of the other one). Indeed, in this case the recursive definition implies that
$$
\#(C(x0)\cup C(x1))=\max(\#C(x0),\#C(x1)).
$$

Assume that Alice has marked one more vertex, some vertex $u$. Then Bob should assign some color to this vertex. The choice of this color will be discussed later; let us see first where (I) may be violated.

\begin{itemize}
\item If $x$ is a descendant of $u$, then (I) remains true, since nothing is changed in the subtree rooted at $x$.
\item If $x$ is incomparable with $u$ (not a prefix and not an extension of $u$), then (I) remains true for the same reason.
\item For $x=u$ both $r(x)$ and $c(x)$ increase by $1$ after marking a vertex and assigning a color to it (Bob has to use a color that did not appear in $C(x)$), so (I) remains true.
\item So the only remaining case is when $x$ is a proper prefix of $u$ (so $u$ is not a root)
\end{itemize}

Let us consider this case in more detail: now it is important which color Bob uses, and we have to prove that he can choose the color in such a way that the invariant remains true.  The problem may appear if at some vertex $x$ (a proper prefix of $u$) the value of $r(x)$ does not change while the value of $c(x)$ changes (increases by $1$ because of the new color). 

We know that $r(u)$ increases by $1$. This increase propagates to the root due to recursive definition. Either it propagates all the way through (and then everything is OK), or the propagation stops at some vertex $v$. This means that we had
$$
r(v)=\max(r(v0),r(v1))\ [+ 1,\text{ if $v$ was marked]},
$$
and one of the two arguments of $\max(\cdot,\cdot)$ increased, but the maximum remained unchanged since the other argument was bigger anyway.

Assume that, say, $r(v0)$ increased ($u$ is in the left subtree of $v$) but $r(v0)=c(v0)$ was less than $r(v1)=c(v1)$, so the maximum did not increase. Then we had $C(v0)\subsetneq C(v1)$ (these sets are comparable and $C(v1)$ is bigger). Then Bob may use the color from $C(v1)\setminus C(v0)$ for the vertex $u$. If he does this, $C(v0)$ increases but remains a subset of $C(v1)$, so $C(v)$ remains unchanged and (I) remains true for $v$ (and for all ancestors of $v$ due to recursive definition).

Summarizing Bob's strategy: when Alice marks some vertex $u$, trace the path from $u$ to the root and look where the marked rank changes (due to the mark at $u$) and where it does not. If it changes all the way to the root (including the root), use whatever color you want. If $v$ is the first vertex where the marked rank remains the same, look at the subtrees rooted at $v0$ and $v1$ and use the color that appears in one of them but not in the other one.

This finishes an alternative proof of Lemma~\ref{lem:gleb}.
\vspace*{-1ex}
\end{proof}
\vspace*{1ex}

Now let us show how the lemma is used to finish the proof of Theorem~\ref{thm:cardinality}. Let $S$ be a function in the class; since $S$ is upper semicomputable, for each $n$ Alice may enumerate strings $x$ such that $S(x)<n$. We know that there is at most $2^n$ strings of this type along any branch of the tree, so Alice never violates the restriction for $k=2^n$. The lemma then says that Bob can assign $2^n$ colors (represented as $n$-bit strings) to all the vertices in such a way that compatible vertices (a string and its prefix) never get the same color. We run these games for all $n$ in parallel; if vertex $x$ gets color $c$, we put $x$ into an enumerable set indexed by $c$. The rules of the game guarantee that all these sets are prefix-free, and the algorithm enumerating $c$th set needs only $|c|$ bits of information. So, if $S(x)<n$, there exists an algorithm of complexity $n+O(1)$ that enumerates a prefix-free set containing $x$. This means that $\KS(x\cnd x*)\le S(x)+O(1)$ as required.
\end{proof}

\subsection{Oracles and the stoppping time complexity}

As every notion in the general computability theory, Kolmogorov complexity can be relativized. Let $X$ be an infinite binary sequence used as an oracle (all the computations get access to $X$ for free). Then we get a notion of \emph{relativized Kolmogorov complexity} $\KS^X(x)$ that can be considered as a function of two arguments, a binary string $x$ and an infinite binary sequence $X$, defined up to $O(1)$ additive term. (The constant in $O(1)$ does not depend on $X$ and $x$.) 

It is natural to compare the stopping time complexity $\KS(x\cnd x*)$ and the relativized complexity $\KS^X(x)$ where $X$ is some oracle (infinite binary sequence) that has $x$ as a prefix.

It is easy to see that 
$$
\KS^X(x)\le \KS(x\cnd x*)
$$ 
for every $X$ that has prefix $x$: an oracle access to entire sequence $X$ is more powerful than a bit-by-bit sequential access to $x$ without the right to read too much (beyond $x$). More formally, let $D$ be a set of triples $(p,x,y)$ used to define $\KS(y\cnd x*)$ (Definition~\ref{def:plain-version}). Then we say that $p$ is a description of $x$ with oracle $X$ (as the definition of $\KS^X(x)$ requires) if $(p,z,x)\in D$ for some $z$ that is a prefix of $X$. For a given $X$ every string $p$ can be a description of only one $x$, since $D$ is monotone. If $(p,x,x)\in D$ and $X$ is an extension of $x$, then $p$ is a description of $x$, and we get the required inequality.

The ``last exit before the bridge'' example shows that $\KS^X(x)$ can be much smaller than $\KS(x\cnd x*)$ for \emph{some} extensions $X$ of $x$: we have $\KS(0^n\cnd 0^n*)=\KS(n)+O(1)$, but $\KS^X(0^n)=O(1)$ for $X=0^n10^\infty$. So it is natural to take \emph{maximum} over all oracles $X$ that extend a given string $x$. Indeed this approach works:

\begin{theorem}\label{thm:stopping-oracle}
$$
\KS(x\cnd x*) = \max \{\KS^X(x)\colon \text{X is an infinite extension of x}\}+O(1).
$$
\end{theorem}

\begin{proof}
As we have already mentioned, $\KS^X(x)\le \KS(x\cnd x*)+O(1)$ for every infinite extension $X$ of $x$. This shows that right hand side does not exceed the left hand side.

To prove the reverse inequality, we use the quantitative characterization of stopping time complexity (Theorem~\ref{thm:cardinality}). Let $S(x)$ be the value of the right hand side. It is enough to prove that $S$ is upper semicomputable and that $S(x)<n$ cannot happen for $2^n$ different prefixes $x$ of some infinite branch~$X$.

The second claim follows directly from the definition. Let $x_1,\ldots,x_k$ be some prefixes of an infinite sequence $X$ such that $S(x_i)<n$ for all $i=1,\ldots,k$. We need to show that $k<2^n$. Since $S(x_i)$ is defined as maximum and $X$ is an extension of $x_i$, we know that $\KS^X(x_i)<n$ for all $i$ and the same $X$. It remains to note that the number of different programs of length less than $n$ is smaller than $2^n$ (and the same programs with the same oracle give the same result).

To show that $S(x)$ is upper semicomputable, we use the standard compactness argument. As usual, it is enough to show that the binary relation $S(x)<n$ is (computably) enumerable. Indeed, for every $x$, the set $\{X\colon \KS^X(x)<n\}$ is the union, taken over all strings $p$ of length less than $n$, of the sets 
$$
\{X\colon \text{$p$ is a description of $x$ with oracle $X$}\}.
$$
Each of these sets is an open set in the Cantor space, since every terminating oracle computation uses only a finite part of the oracle, and the intervals in the Cantor space that form these sets can be effectively enumerated for all $p$ and $x$. The inequality $S(x)<n$ means that the union of these intervals for all $p$ of length less than $n$ covers the Cantor space. Now compactness guarantees that this happens already at some finite stage of the enumeration, so the property $S(x)<n$ is indeed enumerable.
\end{proof}

\section{Non-equivalence results}\label{sec:neg}

\subsection{Prefix-stable or prefix-free functions?}\label{sec:stable-free}

Looking at the characterization of $\KS(x\cnd x*)$ as the minimal complexity of a  program that enumerates a prefix-free set containing $x$ (Theorem~\ref{thm:plain-stopping}), one can ask whether a similar characterization works for the general case, i.e., whether $\KS(y\cnd x*)$ can be characterized as a minimal complexity of programs (machines) with some property. The answer is `yes', but we should be careful while choosing a property of programs used in this characterization. Here are the details.

\begin{definition}
A partial function $f$ defined on binary strings is called 
\begin{itemize}
\item \emph{prefix-free} if its domain is prefix-free (function is never defined on a string and its extension at the same time);
\item \emph{prefix-stable} if for every $x$, if $f(x)$ is defined, then $f$ is defined and has the same value on all (finite) extensions of $x$.
\end{itemize}
\end{definition}

It is easy to see that the definition of $\KS(y\cnd x*)$ can be reformulated in terms of prefix-stable functions: 
\begin{theorem}\label{thm:monot-plain-char}
The minimal plain complexity of a program that computes a prefix-stable function mapping $x$ to $y$ is equal to $\KS(y\cnd x*)+O(1)$.
\end{theorem}
\begin{proof}
A description mode can be considered as a family of prefix-stable functions (indexed by the first argument $p$). This shows that there exist a program for a prefix stable function mapping $x$ to $y$ of complexity at most $\KS(y\cnd x*)+O(1)$. On the other hand, one can efficiently ``trim'' all programs to make them prefix-stable; if $\hat u$ is the trimmed version of a program $u$ and $U$ is the decompressor used to define plain complexity of programs, then the set $D=\{(p,x,\widehat{U(p)}(x))\colon p,x\}$ satisfies the conditions and may be considered as a decompressor in the definition of $\KS(y\cnd x*)$. Using this decompressor, we get the reverse inequality.
\end{proof}

More interesting question: is a similar statement true for prefix-free functions instead of prefix-stable ones? As we mentioned above, Theorem~\ref{thm:plain-stopping} implies that this is the case when $x=y$. (We spoke about programs that stop at $x$, but we may assume without loss of generality that the output is also $x$.) But in the general case it is not true anymore. Let us make this statement more precise. A naive idea is to consider the minimal plain complexity of a program computing a prefix-free function mapping $x$ to $y$. But this quantity does not look reasonable: the complexity of empty string $\Lambda$ with condition $x$  defined in this way is unbounded (and is actually the stopping time complexity of the condition $x$). 

A more reasonable approach is to consider function $\KS'(y\cnd x*)$ defined as the minimal complexity of a prefix-free program that maps \emph{some prefix of $x$} to $y$. This approach still does not work, as the following result shows.
 
 \begin{theorem}\label{thm:stable-not-free}
 The inequality $\KS(y\cnd x*)\le \KS'(y\cnd x*)+c$ holds for some $c$ and for all $x,y$. The reverse inequality does not hold: there exist strings $x_i, y_i$ \textup(for $i=0,1,2,\ldots$\textup) such that $\KS(y_i\cnd x_i*)$ is bounded while $\KS'(y_i\cnd x_i*)$ is unbounded.
 \end{theorem} 

\begin{proof}
The first part is easy: if an algorithm computing a prefix-free function $f$ is given, we can effectively transform it into an algorithm that computes its prefix-stable extension $g$ such that$g(x)=y$ if $f(u)=y$ for some prefix $u$ of $x$.

For the second statement we need to construct a prefix-stable function $F$ that, informally speaking, beats any finite number of prefix-free functions. Let us explain what does it mean. Consider a uniformly computable sequence of all prefix-free functions $G_0,G_1, \ldots$. We need a prefix-stable function $F$ with the following property: \emph{for every $i$ there exist some $x$ and $y$ such that $F(x)=y$ but there is no $j<i$ and no prefix $x'$ of $x$ such that $G_j(x')=y$}. Then we let $x_i$ and $y_i$ be these strings. Since $F(x_i)=y_i$, we know that $\KS(y_i\cnd x_i*)$ is bounded (by complexity of $F$ plus $O(1)$). On the other hand, $\KS'(y_i\cnd x_i*)\to \infty$ as $i\to\infty$ since all programs of bounded complexity appear among $G_0,\ldots,G_{i-1}$ for large enough~$i$.

We define function $F$ step by step, by adding labels to the vertices of the full binary tree. When label $y$ (a binary string) is placed at vertex $x$ (also a binary string), this means that we let $F(x)=y$ and also $F(x')=y$ for all $x'$ that are extensions of $x$ (recall that $F$ should be prefix-stable). There is only one restriction: if $x_1$ and $x_2$ are compatible strings (one is a prefix of the other), and both have labels, these labels should be the same.

We construct $F$ competing with the opponents, as it is often done in algorithmic information theory (see~\cite{game-argument}). There are countably many opponents; $i$th opponent is responsible for $G_i$. We say that she places a label $y$ of color $i$ at vertex $x$ if $G_i(x)$ turns out to be equal to $y$. Note that the opponents' labels carry two types of information: string $y$ and color $i$. Since $G_i$ is prefix-free, $i$th opponent never places her labels at two compatible vertices, so labels on a string and its prefix should never have the same color.

Labels (both placed by us and the opponents) are non-removable. A vertex can have several labels of different colors (corresponding to different opponents) and also our label. The winning condition is formulated for the limit configuration that involves all labels placed during the (infinite) game. We say that opponents $G_0,\ldots,G_{i-1}$ beat us (as a team) if for every label $y$ at vertex $x$ placed by us, there exists $j<i$ and label $y$ of $j$th color placed on $x$ or on some prefix of $x$. If this is not the case, i.e., there exist some label $y$ placed by us at some vertex $x$ such that first $i$ opponents never place label $y$ on $x$ and its prefixes, then we beat first $i$ opponents (as a team). Our goal is to beat all teams (for all $i$).

To achieve this goal, we split the tree into countably many trees $T_i$ as shown (Fig.~\ref{fig:split}); the subtree $T_i$ is used to beat the team $(G_0,\ldots,G_{i-1})$. 
\begin{figure}[h]
\begin{center}
\begin{tikzpicture}[xscale=0.8,yscale=0.8]
\foreach \x/\y in {0/0,-1/1,1/1,0/2,2/2,1/3,3/3}
  \draw[fill] (\x,\y) circle [radius=2pt];
\foreach \x/\y in {3.2/3.2,3.4/3.4,3.6/3.6}
   \draw[fill] (\x,\y) circle [radius=1pt];
\draw (0,0)--(1,1)--(2,2)--(3,3);
\draw (1,1)--(0,2);
\draw (2,2)--(1,3);
\draw (0,0)--(-1,1);
\foreach \x/\y in {-1/1,0/2,1/3}
   \draw [fill=lightgray] (\x,\y)--(\x+0.5,\y+2)--(\x-0.5,\y+2)--cycle;   
 \draw (-1,2.5) node {$T_0$};
  \draw (0,3.5) node {$T_1$};
   \draw (1,4.5) node {$T_2$};
\end{tikzpicture}
\end{center}
\caption{Subtree $T_i$ is used to beat $(G_0,\ldots,G_{i-1})$.}\label{fig:split}
\end{figure}
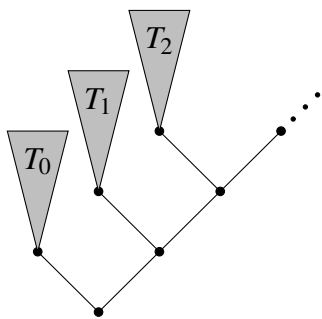
Our strategy considers the trees $T_i$ independently. In this way it is enough to show that we can beat $i$ opponents for any given~$i$. This is done inductively: When constructing a strategy for $T_i$ beating $i$ opponents, we assume that we already know how to beat any smaller number of opponents.

So let us explain the strategy on $T_i$. In this explanation we forget about other subtrees and explain a strategy that beats $i$ opponents on the entire tree. Fix some path in $T_i$, say, the path $1111\ldots$ (Figure~\ref{fig:inside-a-subtree}).
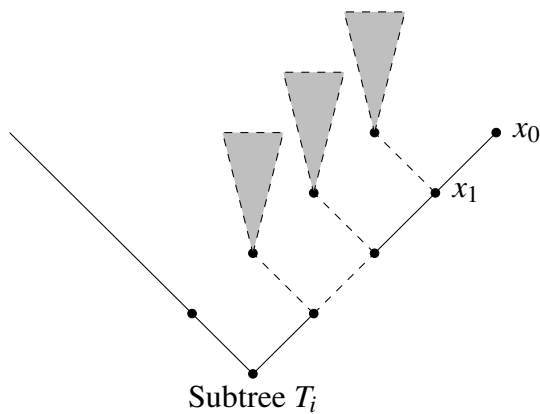
\begin{figure}[h]
\begin{center}
\begin{tikzpicture}[xscale=0.8,yscale=0.8]
\foreach \x/\y in {0/0,-1/1,1/1,0/2,2/2,1/3,3/3,4/4,2/4}
  \draw[fill] (\x,\y) circle [radius=2pt];
\draw (0,0)--(1,1);
\draw [dashed](1,1)--(2,2);
\draw (2,2)--(4,4);
\draw [dashed](3,3)--(2,4);
\draw [dashed](2,2)--(1,3);
\draw [dashed](1,1)--(0,2);
\draw (0,0)--(-4,4);
\foreach \x/\y in {2/4,1/3,0/2}
   \draw [fill=lightgray,dashed] (\x,\y)--(\x+0.5,\y+2)--(\x-0.5,\y+2)--cycle;   
 \draw (0,0) node [below] {Subtree $T_i$};
 \draw (4.1,4) node [right] {$x_0$};
 \draw (3.1,3) node [right] {$x_1$};

\end{tikzpicture}
\end{center}
\caption{Playing inside $T_i$: possible subtrees for playing against $i-1$ opponents.}\label{fig:inside-a-subtree}
\end{figure}
Choose a far enough vertex $x_0=1^m$ on this path\footnote{As we have said, we omit part of the path that is outside $T_i$.} (as we will see, we need $m>i$) and put some fresh (=not used before) label $y$ at vertex $x_0$. Wait until one of the $i$ opponents puts her label $y$ at $x_0$ or on some proper prefix $z$ of $x_0$. If this never happens, we win.\footnote{It may happen also that label $y$ is placed on the path to the root of $T_i$. For our purposes it is the same as if the label $y$ is placed at the root of $T_i$.} There are two possibilities.
\begin{itemize}
\item Some opponent (among the first $i$) places label $y$ at some proper prefix $z$ of $x_0=1^m$ in $T_i$. After that this opponent cannot place any label above $z$, in particular, in the subtree with root $1^{m-1}0$. Then we start to play in this subtree against the remaining $i-1$ opponents using fresh labels (so the opponent who placed $y$ at $z$ is useless for the opponent team). The winning strategy exists due to the inductive assumption.

\item Some opponent (among the first $i$) places label $y$ at $x_0$ itself. Then we place label $y$ at $x_1$ that is the father of $x_0$, i.e., at $x_1=1^{m-1}$, and wait again until some of the first $i$ opponents puts label $y$ at some prefix of $x_1$. If this prefix is a proper prefix of $x_1$, then we know what to do (see above). If the label $y$ is placed at $x_1$ itself, then we place label $y$ at $x_2$ that is the father of $x_1$ (i.e., $x_2=1^{m-2}$) at so on. Finally we either neutralize some opponent, or get labels $y$ at vertices $x_0,x_1,\ldots$, and each of them is replicated (at the same vertex) by one of the opponents, so we get a contradiction at $x_i$ (no more opponents are able to act).
\end{itemize}
   \vspace*{-3ex}
\end{proof}

Theorem~\ref{thm:stable-not-free} implies that the conjecture from~\cite[p.~254]{chernov-et-al} is false, and the function $C_T$ defined there may exceed $C_E$ more than by $O(1)$ additive term. We do not go into the details of the definition used in~\cite{chernov-et-al}; let us mention only that $C_E(y_i\cnd x_i)$ is bounded while $C_T(y_i\cnd x_i)$ is not: for every twice prefix machine (as defined in~\cite[p.~252]{chernov-et-al}) we get a prefix-free function if we fix the first argument (denoted there by $p$).

\subsection{Quantitative characterization of $\KS(y\cnd x*)$ works\\ only up to factor $2$}\label{subsec:quant-general}

In Section~\ref{subsec:quant} we provided a quantitative characterization of stopping time complexity, or $\KS(x\cnd x*)$, with $O(1)$-precision (Theorem~\ref{thm:cardinality}). The natural question is whether a similar characterization can be found in the general case, i.e., for  $\KS(x\cnd y*)$. 

For $\KS(x\cnd y)$ (the standard version of conditional complexity, with no monotonicity requirement) such a characterization is well known: $\KS(x\cnd y)$ is the minimal upper semicomputable function of two arguments $K(x,y)$ such that for every string $y$ and every number $n$ there is at most $2^n$ different strings $x$ such that $K(x,y)<n$.

The natural approach is to keep this restriction and add the monotonicity requirements:
   $$
K(x,y0)\le K(x,y)\ \text{and}\
K(x, y1)\le K(x, y), \ \text{for every $x$ and $y$}.
   $$
We get some class of functions (that are upper semicomputable, satisfy the cardinality restriction and are monotone in the sense described).  Can we characterize $\KS(x\cnd y*)$ as the minimal function in this class? No, as the following theorem shows.

\begin{theorem}\label{thm:no-quant-char}
\leavevmode\par
\textup{(\textbf{a})}~Function~$\KS(x\cnd y*)$ belongs to this class.

\textup{(\textbf{b})}~There exists a minimal \textup(up to $O(1)$ additive term\textup) function in this class;

\textup{(\textbf{c})}~Function~$\KS(x\cnd y*)$ is not minimal in this class: there exist a function $K$ in this class, and sequences of strings $x_n$ and $y_n$ such that $K(x_n,y_n)\le n$, but $\KS(x_n\cnd y_n)\ge 2n-c$ for some $c$ and for every $n$.

\textup{(\textbf{d})}~The factor $2$ that appears in the previous statement is optimal: if $K$ is a function in the class \textup(for example, the minimal one\textup), then $\KS(x\cnd y)\le 2K(x,y)+c$ for some $c$ and for all~$x$ and~$y$.
\end{theorem}

\begin{proof}
The statements (a) and (b) are ``good news'', while the statement (c) is ``bad news'' showing that our characterization does not work. (May be, one can get a natural characterization by adding some other restrictions, but it is quite unclear what kind of restrictions could help here.) Finally, the statement (d) partly saves the situation and shows that the minimal function in the class and $\KS(x\cnd y*)$ differ at most by factor~$2$.

The statement (a) is obvious; note that $\KS(x\cnd y*)$ is bigger than $\KS(x\cnd y)$, so the cardinality restriction remains true. Other requirements immediately follow from the definition.

The statement (b) can be proved in a standard way. We can enumerate all functions in the class and get a uniformly computable sequence of functions $K_m(x,y)$. For that we enumerate all monotone upper semicomputable functions and then ``trim'' them by deleting small values that make the cardinality restriction false. Then we construct the minimal function $K(x,y)$ by letting
$$
 K(x,y)=\min_m K_m(x,y)+m+1.
$$
It is upper semicomputable and monotone; for every $y$, the set of $x$ such that $K(x,y)<n$ is the union of sets $\{x\colon K_m(x,y)<n-m-1\}$ that have cardinality at most $2^{n-m-1}$, and $2^{n-1}+2^{n-2}+\ldots < 2^n$. The function $K$ is minimal, since $K\le K_m+m+1$.

To prove~(c), we need to show that $\KS(x\cnd y*)$ is \emph{not} minimal in the class. We have to construct a function $K(x,y)$ in the class that is smaller than $\KS(x\cnd y*)$. This function will be constructed in the following way. We will make \emph{declarations} of the form ``$K(x,y)\le n$'' for some pairs $(x,y)$ of strings. Each of them implicitly contains declarations ``$K(x,y')\le n$'' for all $y'$ that are descendants (extensions) of $y$. We agree in advance that at most $2^{n-1}$ declarations of this form can be made for each given $y$ (including implicit declarations). Then function $K(x,y)$ is defined as the minimal upper bound declared explicitly or implicitly (for given~$x$ and $y$). For a given $y$, the total number of declarations with upper bound not exceeding $n$ is at most $2^{n-1}+2^{n-2}+\ldots <2^n$, so the function $K$ constructed in this way belongs to our class (we assume that the sequence of declarations is computable; this guarantees that $K$ is upper semicomputable). Note also that the declarations indeed guarantee that the declared inequality is true, if the function $K$ is defined as explained above.

Now we have to describe how $K$ is constructed (how the declarations are made). This is done independently (and in parallel) for each $n$, in some subtree dedicated to $n$. The goal: some declaration $K(x_n,y_n)\le n$ is made for some $x_n$ and $y_n$ such that $\KS(x_n\cnd y_n*)\ge 2n-O(1)$. So we approximate $\KS$ from above, keep track of the changes in these approximations and make declarations trying to beat these changes. (In terms of the game approach one can say that the opponent decreases the complexity and we play against these decreases.) Since we have $n$ fixed, we will not mention $n$ explicitly and read the declaration ``$K(x,y)\le n$'' as ``$x$ is declared simple at $y$'' (implicitly $x$ is declared simple at all descendants of $y$, too). For the changes in the approximations to $\KS(x\cnd y)$ we use a similar language: if $(p,y,x)$ appears in the set $D$ (used to define $\KS(\cdot\cnd\cdot)$), we say that ``description $p$ is allocated to $x$ at $y$'' (where $y$ is considered as a tree vertex). It implies that $p$ is allocated to $x$ in all descendants of $y$, too. Note that it is not possible that the same description is allocated to different objects at the same vertex (but different descriptions may be allocated to the same object).

The restrictions that we have to obey are that \emph{at most $2^{n-1}$ objects can be declared simple at any given vertex} (explicitly or implicitly). Our goal is to guarantee that \emph{some $x$ is declared simple at some vertex $y$, but no description of length $2n-O(1)$ is allocated to $x$ at $y$}. (Here $O(1)$ means some absolute constant that we will fix later; in fact, $6$ will work.)\footnote{Before giving the formal proof, let us say informally what makes the proof possible. Our declarations are ``more flexible'' compared to the actions of our opponent. We need to specify only which objects are simple at a given vertex. The opponent needs to assign specific descriptions to simple objects. These descriptions are inherited in the descendant vertices, and cannot be reused for other objects.} 

\begin{lemma}\label{lem:misha-modified}
By declaring simple at most $2^{n-2}$ objects at each vertex, it is possible either to achieve the goal, or reach a stage when for some vertex $y$ and all its descendants only one object $x$ is declared as simple, but at least $2^{n-2}$ descriptions of length at most $2n-O(1)$ are allocated to~$x$ at $y$.
\end{lemma}

Let us explain why this lemma is enough. We apply it and either achieve the goal, or get some vertex $y$ such that only one object $u$ is declared as simple at $y$ (and in $y$-subtree) but many (at least $2^{n-2}$) descriptions of length $2n-O(1)$ are allocated to~$u$ at~$y$. After that we start the same procedure (guaranteed by Lemma) at $y$-subtree using fresh objects (not~$u$). Note that the descriptions allocated to $u$ at $y$ cannot appear as descriptions of some other objects in $y$-subtree. In this way, using Lemma~\ref{lem:misha-modified} again, we declare at most $2^{n-2}$ simple objects (not counting $u$) at each vertex in $y$-subtree, so the total number of objects declared as simple does not exceed $2^{n-1}$. Lemma guarantees then that either we achieve the goal, or reach a stage where at some vertex $y'$ and its subtree there are two objects declared as simple ($u$ and the newly declared one), and for each of them at least $2^{n-2}$ descriptions of length at most $2n-O(1)$ are allocated at $y'$. Note that allocated descriptions for these two objects are different.

Then we apply Lemma~\ref{lem:misha-modified} third time at the corresponding subtree not using two objects already declared as simple, etc. Finally we may either achieve the goal, or declare up to $2^{n-2}$ objects as simple (at all stages), still obeying the $2^{n-1}$-restriction. For each simple object we have at least $2^{n-2}$ descriptions of length at most $2n-O(1)$, and this is a contradiction (we may use $6$ as a constant in $O(1)$).

\begin{proof}[Proof of the Lemma~\ref{lem:misha-modified}]
Let us first explain how we can achieve the goal or get a vertex~$y$ where only one object is declared as simple, but at least two descriptions of length $2n-O(1)$ are allocated to this object at $y$.

Take some level of a binary tree where we have more vertices than the number of descriptions of the size considered (level $2n$ is OK). At this level declare one simple object per vertex (all objects are different), and wait until a description of the right size is allocated to each of them. Then there are two different vertices $u$ and $u'$ where different objects $z$ and $z'$ are declared as simple, and the same description $p$ is allocated to both (it is OK to use the same description for different objects at different vertices). Then declare $z$ as simple at the root. After that we have only one object $z$ declared as simple at $u$, and two objects declared as simple elsewhere (one declared locally plus $z$). Some description should be allocated to $z$ in the root, and it cannot be $p$, because in this case $p$ would be allocated both to $z'$ and $z$ at $u'$. Therefore two descriptions are allocated to $z$ at $u$.

To amplify this argument and get more descriptions for one object, we use several layers. Consider all vertices of level $2n$ (as used before) and subtrees of height $2n$ rooted at all of them. In each of the subtrees we use the argument above (using disjoint sets of objects) and get an additional description for each subtree root. These descriptions cannot be all different, so there are some objects $z$ and $z'$ declared as simple in two vertices $u$ and $u'$ of height $2n$, and the same description $p$ is allocated to $z$ and $z'$ (at $u$ and $u'$ respectively). Then we declare $z$ as simple at the root of the entire tree, so some description should be allocated to $z$ at the root. It cannot coincide with the descriptions used for $z$ both on levels $2n$ and $4n$, since these descriptions are used for other objects. So we get three descriptions for $z$ at some vertex of level $4n$, and only $z$ is declared there as simple.

We may iterate the argument; the only problem is that the objects declared as simple propagate upwards, so the total number of objects declared as simple increases. So only $2^{n-2}$ iterations are possible, and this gives us the number of descriptions for one object stated in the lemma. Lemma~\ref{lem:misha-modified} is proven.
\end{proof}

This finishes the proof of part~(c). To prove~(d), for a given function $K$ in the class, we construct a description mode $D$ such that $\KS_D(x\cnd y*)\le 2K(x,y)+O(1)$. In fact, both (c) and (d) in fact deal with the same game but provide winning strategies for opposite players, since the game parameters are different. 

We enumerate the function $K(x,y)$ from above. Let us fix some $k$. When we discover that $K(x,y)<k$, we say that \emph{object $x$ is declared simple at vertex $y$}, considering $y$ as a vertex of a full binary tree. This automatically implies that $K(x,y')<k$ for all extensions $y'$ of $y$, so we may assume that when $x$ is declared simple at $y$, it is automatically declared simple at all vertices of $y$-subtree. For every vertex, there is at most $2^k$ objects declared as simple.

Observing this process, we need to construct the description mode $D$. This can be understood as follows: we assign \emph{descriptions} of length $2k+O(1)$ for some objects at some vertices. If description $p$ is assigned to $x$ at $y$, it is automatically assigned to $x$ at all $y'$ that are extensions of $y$. No description should be used for different objects at the same vertex (and, therefore, at a vertex and its extension). Our goal is to provide descriptions (at every vertex) for all objects that are declared simple at that vertex. If we succeed, then this construction can be applied in parallel for all $k$, and we get a description mode $D$ such that $\KS_D(x\cnd y*)\le 2K(x,y)+O(1)$.

It is convenient to denote $2^k$ by $n$. Then we know that at every vertex at most $n$ objects are declared simple, and need to provide descriptions from a pool of size $O(n^2)$ for all simple objects (where the hidden constant does not depend on $n$). How can we achieve this?

We perform the description assignment using several ``layers''. Each layer uses its own pool of descriptions of size $O(n)$. When a new object $x$ is declared as simple at some vertex~$y$, the corresponding request (``please provide a description for $x$ at $y$'') is sent to the first layer, where it is served or rejected. If rejected, the same request is redirected to the second layer, when again it is served or rejected (and redirected to the third layer), etc. We will show that $O(n)$ layers are enough (the requests will never go higher); in total we get $O(n^2)$ descriptions, as required.

All the layers follow that same algorithm of processing requests. The idea is to keep --- as much as possible --- a one-to-one correspondence between objects and descriptions allocated to them. Of course, there is no hope to maintain this correspondence in all situations, since at each layer we have only some maximal number $N=O(n)$ of descriptions, and there are $O(n)$ layers, while the number of objects is unbounded. 

The restricted version of this bijection requirement is as follows:
\begin{quote}
on every path in the tree there is a bijection between the objects served along the path and the descriptions allocated to these objects. Moreover, for every vertex $v$ there exists a bijection between objects served in the $v$-subtree and descriptions allocated to them, \emph{unless there are more than $N$ objects served in the $v$-subtree}.
\end{quote}
Note that:
\begin{itemize}
\item Only requests that reach the layer and are served at this layer are taken into account. Requests that are served by the previous layers, or rejected by our layer (and redirected to the higher layers) do not matter.
\item If a vertex is declared simple at some vertex $v$ and then served, then both the declaration and the description remain valid above $v$ (everywhere in the $v$-subtree).
\item The restriction guarantees that an object never has different descriptions at the same vertex (or at a vertex and its descendant); the same description also cannot be used for different objects at a vertex and its descendant. However, this may happen in two incomparable vertices (and only if there are more than $N$ objects served).
\item If more than $N$ objects are served in the $v$-subtree, then (of course) a bijection between them and descriptions is not possible (for cardinality reasons); the requirement says that this is the only case when the bijection does not exist.
\end{itemize}
Of course, there are easy ways to maintain this invariant relation: just reject all requests, or serve them until $N$ different objects appear and then reject all the subsequent requests. We will describe a better algorithm that serves more requests and guarantees that $O(n)$ layers are enough. Here is it.

We say that (at some stage) a vertex $v$ is \emph{regular} if at most $N$ objects are served in the $v$-subtree (and therefore there is a bijection between objects and descriptions in the $v$-subtree, according to the invariant relation). Otherwise, $v$ is \emph{overloaded}.  We say that object $x$ is \emph{acceptable} at vertex $v$ if $v$ is regular and may remain regular after $x$ is served in $v$ (or above $v$). In other words, $x$ is acceptable at $v$ in two cases: (a)~$v$-subtree has less than $N$ objects (in this case every object is acceptable at $v$); (b)~$v$-subtree has $N$ objects and $x$ is one of them. An observation: if $x$ is acceptable at $v$ and $v'$ is a descendant of $v$, then $x$ is acceptable at $v'$ (since in the $v'$-subtree we have less objects than in $v$-subtree, or the same objects with the same descriptions).

Now the algorithm: when a request to provide a description for an object $x$ at a vertex $v$ arrives, we check whether $x$ is acceptable at $v$. If not, the request is rejected. If yes, we go from $v$ to the root and take the last vertex $\overline v$ where $x$ is acceptable (may be, the root itself). Then we provide a description for $x$ based on the bijection that exists for the $\overline v$-subtree.

\begin{lemma}\label{lem:misha-algorithm}
This algorithm maintains the invariant relation.
\end{lemma}
\begin{proof}
Consider an arbitrary path in the tree. If it does not go through $\overline v$, nothing is changed along the path. If it goes through $\overline v$, the condition along the path remains true, since $\overline v$ was acceptable for $x$ and remains regular.

Now consider an arbitrary vertex $w$. If $w$ is a (proper) ancestor of $\overline v$, then $w$ is now overloaded (because $x$ was not acceptable at $w$). If $w$ is in the $\overline v$-subtree, then the $w$-subtree has a required bijection, since $\overline v$-subtree has it. Finally, if $w$ is incomparable with $\overline v$, then nothing is changed in $w$-subtree. Lemma~\ref{lem:misha-algorithm} is proven.
\end{proof}

It remains to show that $O(n)$ layers are enough if $N$ is large enough; for example, $N=3n$ will work. Here is the main observation. If a request for object $x_0$ at vertex $v_0$ is rejected, this means that $v_0$-subtree already carries at least $N$ objects. They were placed there according to some earlier requests (that were redirected from the previous layers). At most $n$ of these requests can be made for vertices that are on the path from root to $v_0$ (since at most $n$ objects are declared simple or every path). Therefore, at least $2n$ requests were made at vertices that are in the $v_0$-subtrees.  The requests are made for different objects, so we can take one for an object $x_1\ne x_0$, made at some vertex $v_1$ that is in the $v_0$-subtree.

Therefore, on the previous layer a request for $x_1\ne x_0$ was made at $v_1$ that is an extension of $v_0$, and it was rejected. The same reasoning for the previous layer shows that $2n$ earlier requests were accepted at that layer for different objects inside $v_1$-subtree. One of these objects is different from $x_0$ and $x_1$, so some request for an object $x_2\notin\{x_0,x_1\}$ at some vertex $v_2$ in $v_1$-subtree was rejected by the preceding layer, etc.

In this way we get a sequence of different objects $x_0,x_1,\ldots$ requested at vertices $v_0,v_1,\ldots$ where $v_{i+1}$ is an extension of $v_i$. This process continues until we come to the first layer or get $2n$ different objects, and the second case is impossible since all the objects $x_i$ are declared at comparable vertices, and by assumption at most $n$ different objects can be declared simple along a path. Therefore, we come to the first layer in at most $n+O(1)$ steps, so $O(n)$ and even $n+O(1)$ layers are enough. The statement (d) of Theorem~\ref{thm:no-quant-char} is proven.
\end{proof}
\section{Questions}

\begin{question}
Imagine Turing machines with two read-only input tapes; for such a machine $M$ consider a function $f_M$ such that $f_M(x,y)=z$ if $M$ stops at $x$ and $y$ on first and second tape respectively (reading all bits and not more) and produces $z$. Could we characterize the functions $f_M$ (called \emph{twice prefix free} in~\cite[page 242]{chernov-et-al}) or at least their domains? Such a domain is an enumerable set of pairs that does not contain two pairs $(x,y)$ and $(x',y')$ where $x$ is compatible with $x'$ (one is a prefix of the other) and $y$ is compatible with $y'$.  Still this necessary condition is not sufficient, as the following argument shows. Let $z_i$ be a computable sequence of pairwise incompatible strings (say, $z_i=0^i1$). Let $P$ and $Q$ be two enumerable sets that are inseparable (do not have a decidable separating set). Consider the set of pairs that contains
 \begin{itemize}
 \item $(z_i0,z_i0)$ for all $i$;
 \item $(z_i,z_i1)$ for $i\in P$;
 \item $(z_i1,z_i)$ for $i\in Q$.
 \end{itemize}
This set satisfies the necessary condition above (does not contain two compatible pairs). However, assume that some twice prefix free machine has this set as a domain. Then it should terminate after reading $z_i0$ on the first tape and $z_i0$ on the second tape.  Consider the last zero bits on both tapes. One of these bits should be read first (if they are read simultaneously, we may choose any of two). If this is the first bit, then $i\in P$ is impossible (since the machine cannot read $1$ on the second tape before reading $0$ on the first tape). For the same reason, $i\in Q$ is impossible if the second bit is read first. Therefore, a decidable separator exists.
 
 Can we add some conditions to get a characterization of domains of twice prefix free machines?  What do we get if we define stopping time complexity for pairs using machines of this type? Does it have some equivalent description (for example, can it be defined using monotone-conditional complexity with pairs as conditions,   Section~\ref{subsec:monot-cond})?
\end{question}

\begin{question}
Do we have $\KS(x\cnd x*)=\max_z\KS(x\cnd z)+O(1)$ where the maximum is taken over all \emph{finite} extensions $z$ of $x$? (The problem is that the compactness argument does not work anymore.)
\end{question}

\begin{question}
One may consider the function
$$
K(x,y)=\max\{\KS^Y(x)\colon \text{$Y$ is an infinite extension of $y$}\}
$$ 
We have shown that for $x=y$ it coincides with $\KS(x\cnd x*)$, showing that it does not exceed $\KS(x\cnd x*)$ and satisfies the quantitative restrictions of Theorem~\ref{thm:cardinality}. Both arguments remain valid (with minimal changes) for the general case, and we conclude that $K(x,y)$ defined in this way does not exceed $\KS(x\cnd y*)$ and also satisfied the cardinality restrictions  of Theorem~\ref{thm:no-quant-char}, (b). However, now these upper bound and lower bound differ, and we do not know where between them the function $K(x,y)$ lies. Does it coincide with its upper bound $\KS(x\cnd y*)$ for arbitrary $x$ and $y$, or with its lower bound, the minimal upper semicomputable function that satisfies the cardinality requirements (see Theorem~\ref{thm:no-quant-char}), or neither?
\end{question}

\subsection*{Acknowledgments}
The authors are grateful to Alexey Chernov, Volodya Vovk, members of the ESCAPE team (LIRMM, Montpellier), Kolmogorov seminar (Moscow) and Theoretical Computer Science Laboratory (National Research University Higher School of Economics, Computer Science department), and the participants of Dagstuhl meeting  where some results of the paper were presented~\cite{short}.

\end{document}